\documentclass[journal,twocolumn]{IEEEtran}

\usepackage{epsfig}
\usepackage{subfigure}
\usepackage{graphicx}

\usepackage{amsmath}
\usepackage{amssymb}

\def\BibTeX{{\rm B\kern-.05em{\sc i\kern-.025em b}\kern-.08em
    T\kern-.1667em\lower.7ex\hbox{E}\kern-.125emX}}

\newtheorem{claim}{Claim}

\newtheorem{theorem}{Theorem}

\newtheorem{remark}{Remark}

\title{Two-way Interference Channels}

\author{Changho Suh, I-Hsiang Wang and David Tse \\
\thanks{C. Suh is with the Research Laboratory of Electronics at Massachusetts Institute of Technology, Cambridge, USA (Email: $\mathsf{chsuh@mit.edu}$)}
\thanks{I. Wang is with the School of Computer and Communication Sciences, Ecole Polytechnique F\'ed\'erale (EPFL), Lausanne, Switzerland (Email: $\mathsf{i\textrm{-}hsiang.wang@epfl.ch}$)}
\thanks{D. Tse is with the Wireless Foundations in the University of California at Berkeley, Berkeley, USA (Email: $\mathsf{dtse@eecs.berkeley.edu}$).} }

\begin{document}

\maketitle

\begin{abstract}
We consider two-way interference channels (ICs) where forward and backward channels are ICs but not necessarily the same. We first consider a scenario where there are only two forward messages and feedback is offered through the backward IC for aiding forward-message transmission. For a linear deterministic model of this channel, we develop inner and outer bounds that match for a wide range of channel parameters.  We find that the backward IC can be more efficiently used for feedback rather than if it were used for sending its own independent backward messages. As a consequence, we show that feedback can provide a \emph{net} increase in capacity even if feedback cost is taken into consideration. Moreover we extend this to a more general scenario with two additional independent backward messages, from which we find that interaction can provide an arbitrarily large gain in capacity.
\end{abstract}

\section{Introduction}

The inherent two-way nature of communication links allows nodes to adapt their transmitted signals to the past received signals in exchanging their messages. Understanding the role of interaction
lies at the heart of two-way communication.
However, even for the point-to-point two-way communication first addressed by Shannon~\cite{shannon:two-way}, we are still lacking in our understanding of how to treat two-way information exchanges, and the underlying difficulty has impeded progress on this field over the past few decades.

Since interaction is enabled through the use of feedback, \emph{feedback} is a more basic research topic that needs to be explored towards understanding two-way communication.
The history of feedback traces back to Shannon~\cite{shannon:it} who showed that feedback provides no gain in capacity for discrete memoryless point-to-point channels.
Although feedback can indeed increase the capacity of multiple access channels~\cite{Gaarder:it,Cover:it81}, the increase in capacity for the Gaussian case is bounded by 1 bit for all channel parameters~\cite{Ozarow:it}.
Due to these results, traditionally it is believed that interaction has had little impact on increasing capacity.

In contrast, recent research shows that feedback provides more significant gain for communication over interference channels (ICs)~\cite{Kramer:it02,SuhTse}. Interestingly the feedback gain is shown to be unbounded, i.e., the gap between the feedback and nonfeedback capacities can be arbitrarily large for certain channel parameters. This result motivates us to challenge system implementation, since it relies on an idealistic scenario where perfect feedback is given for free. So a natural question that arises is to ask whether feedback can provide a \emph{net} increase in capacity even if feedback cost is taken into consideration. The first attempt to address this question has been made in~\cite{AlirezaSuhAves} where it was shown that one bit of feedback is worth at most one bit of capacity, when feedback links are modeled as rate-limited bit pipes. This implies that there is no net feedback gain unless feedback cost is cheaper.

However, this result does not well evaluate the net feedback gain in a more realistic scenario where feedback is provided through the backward IC, not through bit pipes. This motivates us to explore two-way ICs where both forward and backward channels are ICs but not necessarily the same.\footnote{In FDD systems, the forward and backward channels are on completely different bands. In TDD systems, those channels can be on different subcarriers or different coherence times.}
We first consider a simple scenario where there are only two forward messages and feedback is offered through the backward IC for helping forward-message transmission. To capture possibly different symbol rates between forward and backward channels,\footnote{In 3GPP-LTE and WiMAX systems employing an OFDM modulator, only a few subcarriers are assigned for feedback, which incurs different forward and backward (feedback) symbol rates.} we introduce a parameter $\lambda$ which indicates the fraction of time that the backward channel uses for feedback. The remaining $(1-\lambda)$ fraction of time can be used for other purpose, e.g., sending two independent backward messages.


For the Avestimehr-Diggavi-Tse (ADT) deterministic model~\cite{Salman:IT11} of this channel, we develop inner and outer bounds on the feedback capacity that match for a wide range of channel parameters.
As a result, we find that the capacity gain due to feedback can be strictly larger than the capacity gain due to the use of the backward IC for sending independent backward messages. In other words, the backward IC can be more efficiently used for sending feedback signals that aid forward-message transmission, rather than if it were used for transmitting its own independent backward messages. This finding shows that feedback can provide a net gain in capacity even if we take feedback cost into consideration, i.e., subtract the capacity gain due to the independent-backward-message transmission. The gain comes from the fact that the backward IC's use for feedback enables the exploitation of side information at forward-message-senders to make the backward IC effectively more capable. We also extend this idea to a more general scenario where there are two additional independent backward messages. As a consequence, we show that interaction can provide an arbitrarily large gain in capacity. Moreover we find that this gain can be larger when allowing the mixture of forward-and-backward messages for transmission.



{\bf Related Work:} 
In~\cite{Sahai:ITW09}, Sahai~\emph{et.al.} also considered the two-way IC with two forward messages, and they showed that there is no net feedback gain when forward-and-backward channels are the same and lie in the strong interference regime. On the other hand, we consider arbitrary forward and backward channels, and find that feedback can provide a net capacity gain for some channel regimes. In Section~\ref{section:Comparison}, we will provide details on this.

\section{Model}

Fig.~\ref{fig:generalfeedbackmodel} describes a two-way ADT deterministic IC where user $k$ wants to send its own message $W_k$ to user $\tilde{k}$, and user $\tilde{k}$ feeds back a function of its received signal over the backward IC during the $\lambda$ fraction of time, $k=1,2$. We assume that $W_1$ and $W_2$ are independent and uniformly distributed. For simplicity, we consider a setting where both forward and backward ICs are symmetric but not necessarily the same. In the forward IC, $n$ and $m$ indicate the number of signal bit levels for direct and cross links respectively. The corresponding values in the backward IC are denoted by $(\tilde{n}, \tilde{m})$. Let $X_k \in \mathbb{F}_2^{\max(n,m)}$ be user $k$'s transmitted signal and $V_{k} \in \mathbb{F}_2^m$ be a part of $X_k$ visible to user $\tilde{j} (\neq \tilde{k})$. Similarly let $\tilde{X}_k$ be user $\tilde{k}$'s transmitted signal and $\tilde{V}_k$ be a part of $\tilde{X}_k$ visible to user $j (\neq k)$.

\begin{figure}[t]
\begin{center}
{\epsfig{figure=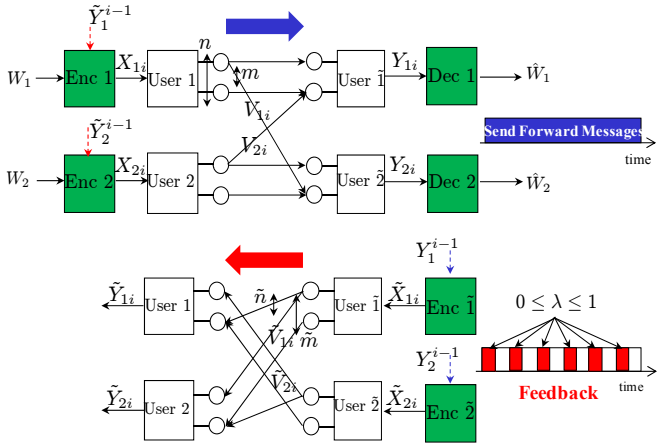, angle=0, width=0.48\textwidth}}
\end{center}
\caption{Two-way ADT deterministic interference channel (IC).} \label{fig:generalfeedbackmodel}
\end{figure}

The encoded signal $X_{ki}$ of user $k$ at time $i$ is a function of its own message and past feedback signals: $X_{ki}= f_{ki} (W_k, \tilde{Y}_{k}^{i-1})$.
We define $\tilde{Y}_{k}^{i-1}:= \{ \tilde{Y}_{kt} \}_{t=1}^{i-1}$ where $\tilde{Y}_{kt}$ denotes the feedback signal received at user $k$ at time $t$. User $\tilde{k}$'s transmitted signal $\tilde{X}_{ki}$ is a function of its past output sequences:
$\tilde{X}_{ki} = \tilde{f}_{ki} ({Y}_{k}^{i-1})$.
We assume that the $\lambda$ fraction of time is assigned to the backward channel to use for feedback. This induces $\sum_{i=1}^{N} H(\tilde{X}_{ki}) \leq N \lambda \max( \tilde{n}, \tilde{m} )$, where $N$ indicates code length. We define  $\tilde{X}_{k}^N$ as a whole vector that includes feedback signals as well as null signals (mapping to no feedback transmission), e.g., $\tilde{X}_{k}^N = \{\varnothing,\tilde{X}_{k2}, \varnothing, \tilde{X}_{k4}, \cdots \}$. A rate pair $(R_1,R_2)$ is said to be achievable if there exists a family of codebooks and  encoder/decoder functions such that the average decoding error probabilities go to zero as code length $N$ tends to infinity. The capacity region ${\cal C}$ is the closure of the set of achievable rate pairs. The sum capacity is defined as $C_{\sf sum} = \sup \left\{ R_1 + R_2: (R_1,R_2) \in \mathcal{C} \right\}$.

\section{Main Results}
\label{sec:MainResults}

\begin{theorem}[Achievability]
\label{thm:achievability}
Let $\alpha=\frac{m}{n}$.
\begin{align*}
R_{\sf sum} = \left\{
                \begin{array}{ll}
                \min \left\{ C_{\sf no} + 2\lambda \tilde{n}, C_{\sf pf} \right \},&\hbox{$\alpha \geq 2$,} \\
 \min\left\{ C_{\sf no} + 2 \lambda  \max (\tilde{n} - \tilde{m}, \tilde{m}), C_{\sf pf}\right \},& \hbox{$\alpha < \frac{2}{3}$,} \\
               C_{\sf no},& \hbox{o.w.}
              \end{array}
              \right.
\end{align*}
where $C_{\sf no}$ and $C_{\sf pf}$ indicate the nonfeedback and perfect-feedback sum capacities respectively~\cite{ElGamal:it82,bresler:europe,SuhTse}:
\begin{align*}
&C_{\sf no}=
\left\{
                \begin{array}{ll}
                2n,&\hbox{$\alpha \geq 2$,} \\
                2 \max (n-m, m),& \hbox{$\alpha < \frac{2}{3}$,} \\
              \max( 2n - m, m),& \hbox{o.w.}
              \end{array}
              \right.
\\
&C_{\sf pf}= \max( 2n - m, m).
\end{align*}
\end{theorem}
\begin{proof}
See Section~\ref{sec:achievability}.
\end{proof}


\begin{theorem}[Outer Bound]
\label{thm:outerbound}
\begin{align*}
C_{\sf sum} \leq \left\{
                \begin{array}{ll}
                \min \left\{ C_{\sf no} + 2\lambda \tilde{n}, C_{\sf pf} \right \},&\hbox{$\alpha \geq 2$,} \\
 \min\left\{ C_{\sf no} + 2 \lambda  \max (\tilde{n}, \tilde{m}), C_{\sf pf}\right \},& \hbox{$\alpha < \frac{2}{3}$,} \\
               C_{\sf no},& \hbox{o.w.}
              \end{array}
              \right.
\end{align*}
\end{theorem}
\begin{proof}
See Section~\ref{sec:outerbound}.
\end{proof}
\begin{theorem}[Sum Capacity]
\label{thm:sumcapacity}
The inner bound and the outer bound (given in Theorems~\ref{thm:achievability} and $\ref{thm:outerbound}$ respectively) match and thus establish the sum capacity, except for the regime of $\left( \alpha < \frac{2}{3}, \tilde{\alpha} < 1 \right)$, where $\tilde{\alpha}:=\frac{\tilde{m}}{\tilde{n}}$.
\end{theorem}
\begin{proof}
The proof is immediate. Note that the inner and outer bounds differ only when $\alpha < \frac{2}{3}$. The inner bound contains $2 \lambda \max (\tilde{n}-\tilde{m}, \tilde{m})$, while the outer bound has $2 \lambda \max (\tilde{n}, \tilde{m})$. These two terms coincide if $\tilde{\alpha} \geq 1$; differ otherwise.
\end{proof}

\begin{figure}[t]
\begin{center}
{\epsfig{figure=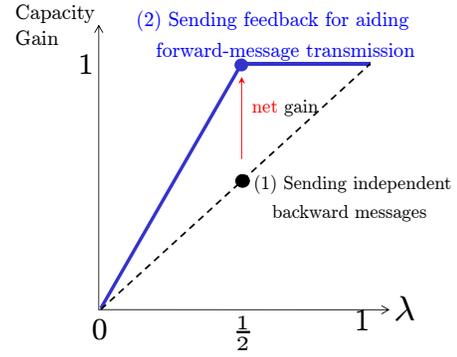, angle=0, width=0.32\textwidth}}
\end{center}
\caption{Net feedback gain: $(n,m)=(2,1)$ and $(\tilde{n},\tilde{m})=(1,1)$.
} \label{fig:netfeedbackgainstrong}
\end{figure}

{\bf Net Feedback Gain:}
Note that in two-way communication, there are two ways of using the backward IC: (1) Sending independent backward messages; (2) Sending feedback signals to help forward-message transmission. Using the above theorems, we will now explain why the backward IC can be more efficiently used for the second purpose, rather than if it were used for the first purpose. Consider an example where $(n,m)=(2,1)$ and $(\tilde{n},\tilde{m})=(1,1)$. Suppose that the $\lambda$ fraction of time is assigned for sending independent backward messages. The capacity gain offered by the backward IC is
\begin{align}
\label{eq:delC1}
\Delta C_{\sf sum} = \lambda C_{\sf B} = \lambda \textrm{ bits},
\end{align}
where $C_{\sf B}$ denotes the nonfeedback sum capacity of the backward IC. In this example, $C_{\sf B}=1$. Suppose that the $\lambda$ fraction of time is now assigned for feedback. Then, due to Theorems~\ref{thm:achievability} and $\ref{thm:sumcapacity}$, the capacity gain offered by the backward IC is
\begin{align}
\label{eq:delC2}
\Delta C_{\sf sum} &=  \min \left\{ 2 \lambda \max (\tilde{n}-\tilde{m}, \tilde{m}), \min (m, 2n-3m) \right \} \nonumber \\
&= \min \left\{ 2\lambda, 1 \right\} \textrm{ bits}.
\end{align}
Fig.~\ref{fig:netfeedbackgainstrong} plots these two capacity gains as a function of $\lambda$. Notice that when $\lambda=0.5$, the capacity gain due to the first purpose is $0.5$ bits; on the other hand, the capacity gain due to the second purpose is $1$ bit.
Without feedback cost, the capacity gain due to feedback is 1 bit. Taking the feedback cost into consideration, we now subtract the capacity gain due to the first purpose; hence, a net gain in capacity is $1 -0.5=0.5$ bits. This implies net feedback gain.

\section{Proof of Achievability}
\label{sec:achievability}

{\bf Review of the Perfect Feedback Scheme~\cite{SuhTse}:} Let us start by examining the perfect feedback scheme. In the very strong interference regime of $\alpha \geq 2$, feedback creates a better alternative path, e.g., $[User1\rightarrow User\tilde{2} \rightarrow \textrm{feedback} \rightarrow User2 \rightarrow User\tilde{1}]$, thus enabling each user to relay the other user's information.
In the regime of $\alpha < \frac{2}{3}$, information is split into common and private parts.
Feedback allows each user to decode other user's common information to forward it later.
This forwarding enables user $\tilde{k}$ to refine the corrupted information in the past while not interfering with the other user's transmission. Here the key observation is that perfect feedback enables each user to decode the other user's common information. Our model, however, provides feedback in the limited fashion, rendering the decoding operation challenging. We will next show how to overcome this challenge. Our feedback strategy is categorized into three types depending on the values of $(\alpha,\tilde{\alpha})$.

\subsection{Type I: $\alpha \geq 2$}

\begin{figure}[t]
\begin{center}
{\epsfig{figure=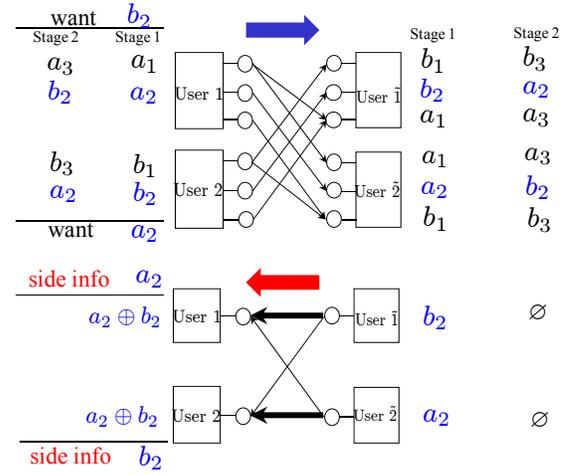, angle=0, width=0.4\textwidth}}
\end{center}
\caption{Type I achievable scheme for $\alpha:=\frac{m}{n}=3$,  $\tilde{\alpha}:=\frac{\tilde{m}}{\tilde{n}}=1$ and $\lambda = \frac{1}{2}$. When \emph{side information} is exploited at user 1 and 2, the backward IC can be cast as two non-interfering point-to-point channels. This enables \emph{two}-bit transmission of $(b_2, a_2)$ over \emph{one}-bit-capacity backward IC, thus providing net feedback gain.
} \label{fig:strong}
\end{figure}

Let us explain the scheme with an example in Fig.~\ref{fig:strong}, where $(n,m)=(1,3)$, $(\tilde{n},\tilde{m})=(1,1)$ and $\lambda = \frac{1}{2}$. Similar to the perfect feedback scheme, it has two stages. In the first stage, each user starts with sending $C_{\sf no}/2=n$ bits on the upper levels. In this example, each user sends 1 bit. On the next lower level, it sends an additional bit; as a result, user 1 and 2 send $(a_1,a_2)$ and $(b_1,b_2)$ respectively. User $\tilde{1}$ then gets $a_1$ while receiving $(b_1,b_2)$ from user 2. Similarly user $\tilde{2}$ gets $b_1$ and $(a_1,a_2)$.

Similar to the perfect feedback scheme, user $\tilde{k}$ feeds back the other user's information (not received yet at the desired place). But the difference here is that this transmission is through the backward IC. Suppose that user $\tilde{1}$ and $\tilde{2}$ simultaneously send $b_2$ and $a_2$ respectively.
Then, it seems impossible to decode these two bits, since each user receives the same signal. It seems that two time slots are needed to feed back these two bits. However, we can actually accomplish this in one shot. The idea is to exploit side information. Exploiting $a_2$ as side information, user 1 can decode $b_2$, and similarly user 2 can decode $a_2$. Here the key observation is that with \emph{side information} at user 1 and 2, the backward IC can be viewed as \emph{two non-interfering point-to-point channels}. In the second stage, each user sends its own fresh information on the first level and forwards the other user's information on the second level: user 1 and 2 send $(a_3,b_2)$ and $(b_3,a_2)$ respectively. User $\tilde{1}$ can then decode its own fresh information $a_3$ as well as $a_2$ which was not received in the first stage. Similarly user $\tilde{2}$ can decode $(b_3,b_2)$. Therefore, we can achieve the sum rate of $3$, showing a $50\%$ improvement from the nonfeedback capacity of 2. Note that the backward channel is utilized once every two slots and therefore we satisfy the constraint of $\lambda =\frac{1}{2}$.

We now extend this to arbitrary values of $(n,m)$, $(\tilde{n},\tilde{m})$ and $\lambda$. In the first stage, each user starts with sending $C_{\sf no}/2=n$ bits on the upper levels. On the next lower levels, it sends the following number of additional bits:
\begin{align}
\label{eq:feedbackbit_strong}
\min  \{ 2 \lambda \tilde{n},  C_{\sf fb}-C_{\sf no}  \},
\end{align}
where $C_{\sf fb}-C_{\sf no} = m-2n$ in the regime $\alpha \geq 2$. Notice that the maximum number of bits that can be squeezed in is limited by $C_{\sf fb}-C_{\sf no}$.
Recall from the above example that the backward IC can be cast into two non-interfering point-to-point channels. So the effective capacity of the backward IC per user for the purpose of feedback is $\tilde{n}$. We multiply this by 2, as two stages are employed.\footnote{Here $2 \lambda \tilde{n}$ can be a non-integer rational number, which is incompatible with an input of the ADT model. However, we can resolve this by employing multiple time slots, say $M$, within each stage, since $2 \lambda \tilde{n} M$ can be made an integer.}
Through the backward IC, user $\tilde{k}$ can now relay the amount $\min  \{ 2 \lambda \tilde{n},  C_{\sf fb}-C_{\sf no}  \}$ of the other user's information. In the second stage, each user sends $C_{\sf no}/2$ fresh bits on the upper levels and the other user information (decoded with feedback) on the next lower levels. User $\tilde{k}$ can then decode $C_{\sf no} + \min \left \{ 2 \lambda \tilde{n},  C_{\sf fb}-C_{\sf no} \right \}$ during the two stages. Therefore, we can achieve $R_{\sf sum}$ in Theorem~\ref{thm:achievability}.

\begin{remark}[Exploiting Side Information]
Note in Fig.~\ref{fig:strong} that the two bits $(b_2,a_2)$ can be fed back through the one-bit-capacity backward IC. This is because each user can cancel the seemingly interfering information by exploiting its own information as side information. This enables the net feedback gain: a capacity increase of 1 bit with $\lambda C_{\sf B}=\frac{1}{2}$ bits of the backward IC's original capability. The nature of the feedback gain offered by side information coincides with that of the butterfly example~\cite{ahlswede:it} and many other network-coding examples~\cite{Wu:05,Katti:SIGCOMM06,BarYossef:FOCS06, SuhTse, Tassiulas:NetCom09,MaddahAli:allterton10}. $\Box$
\end{remark}

\subsection{Type II: $\alpha < \frac{2}{3}, \tilde{\alpha} \geq \frac{1}{2}$}

\begin{figure}[t]
\begin{center}
{\epsfig{figure=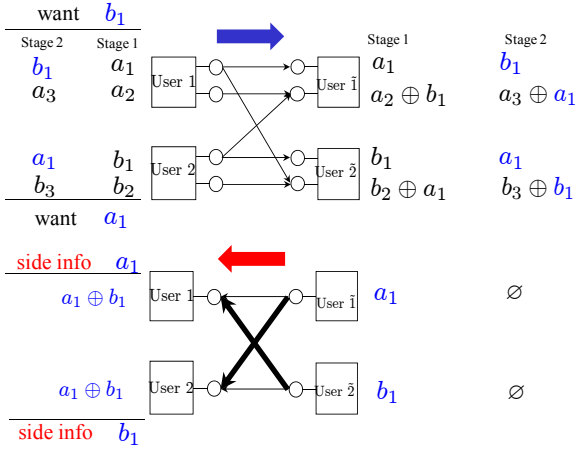, angle=0, width=0.42\textwidth}}
\end{center}
\caption{Type II achievable scheme for $\alpha:=\frac{m}{n}=\frac{1}{2}$,  $\tilde{\alpha}:=\frac{\tilde{m}}{\tilde{n}}=1$ and $\lambda = \frac{1}{2}$.
} \label{fig:weak1}
\end{figure}

We explain the second-type scheme using an example illustrated in Fig.~\ref{fig:weak1}. Here $(n,m)=(2,1)$, $(\tilde{n},\tilde{m})=(1,1)$ and $\lambda = \frac{1}{2}$. Similar to Type I, it has two stages. In the first stage, each user starts with sending $C_{\sf no}/2 = \max (n-m,m)$ bits, comprised of $(n-m)$ private bits and $(2m - n)^+$ common bits. In this example, each user sends one private bit only. On the upper common levels, each user sends the following number of additional bits:
\begin{align}
\label{eq:feedbackbit_weak1}
\min  \{ 2 \lambda \tilde{m}, C_{\sf pf} - C_{\sf no}  \},
\end{align}
where $C_{\sf pf} - C_{\sf no} = \min (m, 2n-3m)$ in the regime $\alpha < \frac{2}{3}$.
In the sequel, we will show that in this regime the backward IC can be viewed as two non-interfering \emph{cross} point-to-point channels, thus making the effective capacity of the backward IC per user $\tilde{m}$. So we have used the $2\lambda \tilde{m}$ for the number of additional bits. Similarly, a factor of 2 is multiplied due to the two-stage nature of the scheme. In this example, $2\lambda \tilde{m}=1$. So user 1 and 2 send $a_1$ and $b_1$ respectively.

In the perfect feedback scheme, user 1 wanted to know the other user information $b_1$ which caused interference to its desired symbol $a_2$. Similarly user 2 wanted to know $a_1$. And to satisfy this demand, the two bits $(a_2 \oplus b_1, b_2 \oplus a_1)$ were fed back to the users. Suppose we mimic this transmission: user $\tilde{1}$ and $\tilde{2}$ send $a_2 \oplus b_1$ and $b_2 \oplus a_1$ respectively. Unfortunately this does not work. User 1 cannot decode $b_1$ from the received signal $a_1\oplus a_2 \oplus b_1 \oplus b_2$ and similarly user 2 cannot decode $a_1$. It seems that two time slots are needed to feed back these two bits. However, we can satisfy the demand in one shot.
Note that the symbol $b_1$ wanted by user 1 is available at user $\tilde{2}$. Similarly the symbol $a_1$ wanted by user 2 is available at user $\tilde{1}$. Suppose we now send these two bits instead. User 1 and 2 can then decode $b_1$ and $a_1$ respectively, exploiting its own signal as side information. The key observation here is that exploiting side information at user 1 and 2, the backward IC becomes equivalent to two non-interfering cross point-to-point channels.
This enables feeding back the following number of bits: $\min  \{ 2 \lambda \tilde{m}, C_{\sf pf} - C_{\sf no}  \}$. In the second stage, each user starts with the nonfeedback scheme and additionally sends the other user's information (decoded with feedback) on vacant common levels. User $\tilde{k}$ can then decode $C_{\sf no} + \min \left \{ 2 \lambda \tilde{m},  C_{\sf fb}-C_{\sf no} \right \}$ bits during the two stages, thus achieving $R_{\sf sum}$ in Theorem~\ref{thm:achievability}.

\subsection{Type III: $\alpha < \frac{2}{3}, \tilde{\alpha} < \frac{1}{2}$}


The only distinction with respect to Type II is that in this regime of $\tilde{\alpha} <\frac{1}{2}$, the effective capacity of the backward IC per user for the purpose of feedback is now $\tilde{n}-\tilde{m}$. This is because the backward IC is now equivalent to two point-to-point channels \emph{composed of private levels only}.
In Fig.~\ref{fig:weak1}, remember that user $\tilde{1}$ fed back $a_1$ to user 2 through the \emph{cross} link, so the transmission rate was limited by $\tilde{m}$. However, in the regime of $\tilde{\alpha} < \frac{1}{2}$, $\tilde{n}- \tilde{m} > \tilde{m}$. This motivates us to consider a better alternative: user $\tilde{k}$ uses $\tilde{n}-\tilde{m}$ private levels for feedback. For example, user $\tilde{1}$ can alternatively feed back $a_2 \oplus b_1$ using a private level, thus allowing user 1 to decode $b_1$ as $a_2$ can be subtracted. Also this private-level transmission does not hurt the other-link transmission.
  Taking this alternative, the effective capacity of the backward IC per user for the purpose of feedback is $\tilde{n}-\tilde{m}$. This way, user $\tilde{k}$ can decode $C_{\sf no} + \min \left \{ 2 \lambda (\tilde{n}-\tilde{m}),  C_{\sf fb}-C_{\sf no} \right \}$ bits during the two stages. This completes the proof.

\section{Proof of Outer Bound}
\label{sec:outerbound}

One can see that it suffices to prove the following bounds:
\begin{align}
\label{eq:cutsetbound} C_{\sf sum} \leq & \min  \left\{  2n + 2\lambda \tilde{n}, \right. \\
\label{eq:perfectfeedbackbound}  &\qquad \left. (n-m)^+ + \max(n,m), \right. \\
\label{eq:newbound}  & \qquad \left. 2 \max (n-m,m) + 2 \lambda \max ( \tilde{n}, \tilde{m}) \right \}.
\end{align}
The proof of the bound (\ref{eq:cutsetbound}) is based on the standard cutset argument. Note that the bound (\ref{eq:perfectfeedbackbound}) matches the perfect-feedback bound~\cite{SuhTse,Sahai:ITW09}; hence it is also an outer bound of our channel. For completeness, we include these proofs in the following subsections. The main focus of this section is to prove the last bound (\ref{eq:newbound}).

\subsection{Proof of~(\ref{eq:cutsetbound})}
Starting with Fano's inequality, we get
\begin{align*}
\begin{split}
N&(R_1 - \epsilon_N) \leq I(W_1;Y_1^{N}, \tilde{Y}_2^N, W_2) \\
& \overset{(a)}{=} \sum H(Y_{1i}, \tilde{Y}_{2i} | W_2, Y_1^{i-1}, \tilde{Y}_2^{i-1}, X_{2i}) \\
& \overset{(b)}{=} \sum H(Y_{1i} | W_2, Y_1^{i-1}, \tilde{Y}_2^{i-1}, X_{2i}) \\
& \quad + \sum H(\tilde{Y}_{2i} | W_2, Y_1^{i}, \tilde{Y}_2^{i-1}, X_{2i}, \tilde{X}_{1i}) \\
&\overset{(c)}{\leq} \sum H(Y_{1i} | X_{2i} )  + \sum H(\tilde{Y}_{2i} | \tilde{X}_{1i}) \\
& \overset{(d)}{\leq}  N (n + \lambda \tilde{n})
\end{split}
\end{align*}
where $(a)$ follows from the fact that $W_1$ is independent of $W_2$, and $X_{2i}$ is a function of $(W_2,\tilde{Y}_2^{i-1})$; $(b)$ follows from the fact that $\tilde{X}_{1i}$ is a function of $Y_1^{i-1}$; $(c)$ follows from the fact that conditioning reduces entropy; $(d)$ follows from the fact that the right-hand-side is maximized when $(X_1, X_2, \tilde{X}_1, \tilde{X}_2)$ are uniform and independent, and $\sum H(\tilde{Y}_{2i}|\tilde{X}_{1i}) \leq N \lambda \tilde{n}$. Similarly we can show $N(R_2 - \epsilon_N) \leq N( n + \lambda \tilde{n})$. If $(R_1,R_2)$ is achievable, then $\epsilon_N \rightarrow 0$ as $N$ tends to infinity. Therefore, we get the desired bound.

\subsection{Proof of~(\ref{eq:perfectfeedbackbound})}
Starting with Fano's inequality, we get
\begin{align*}
\begin{split}
&N(R_1 + R_2- \epsilon_N) \overset{(a)}{ \leq} I(W_1;Y_1^{N}| W_2) + I(W_2;Y_2^{N} )  \\
& = H(Y_1^{N}| W_2) + H(Y_2^{N})  \\
&\quad  -\left \{  H(Y_1^N, Y_2^N | W_2) -  H(Y_1^{N}|  W_2, Y_2^N) \right \} \\
& = H(Y_1^{N}|W_2, Y_2^N) -  H(Y_2^{N}|  W_2, Y_1^N) + H(Y_2^N ) \\
& \leq H(Y_1^{N}|W_2, Y_2^N) + H(Y_2^N ) \\
& \overset{(b)}{=} \sum H(Y_{1i} | W_2, Y_2^N, Y_1^{i-1},  \tilde{X}_1^i, \tilde{X}_{2i}, \tilde{Y}_2^i, X_{2i}, V_{1i}) + H(Y_2^N )  \\
& \overset{(c)}{\leq}  \sum H(Y_{1i} | V_{1i}, X_{2i} ) + \sum H(Y_{2i}) \\
& \leq N \left\{ (n-m)^+ + \max(n,m) \right \}
\end{split}
\end{align*}
where $(a)$ follows from the independence of $(W_1,W_2)$; $(b)$ follows from the fact that $\tilde{X}_{1}^{i}$ is a function of $Y_1^{i-1}$, $X_{2i}$ is a function of $(W_2, \tilde{Y}_2^{i-1} )$, and $V_{1i}$ is a function of $(X_{2i}, Y_{2i})$; $(c)$ follows from the fact that conditioning reduces entropy. This completes the proof.

\subsection{Proof of~(\ref{eq:newbound})}

Starting with Fano's inequality, we get
\begin{align*}
\begin{split}
&N(R_1 + R_2- \epsilon_N) \leq I(W_1;Y_1^{N} ) + I(W_2;Y_2^{N})  \\
&= H(Y_1^{N}) - H(Y_1^N|W_1)  + H(Y_2^{N}) -H(Y_2^{N}|W_2) \\
& \overset{(a)}{=}    H(Y_1^{N},V_1^N )  - H( V_1^N )  +  H(Y_2^{N}, V_2^N)  - H( V_2^N ) \\
&  + \left\{ H (V_1^N) - H( Y_1^N | W_1 ) -  H( V_1^N | Y_1^N)  \right \} \\
& + \left\{ H (V_2^N) - H( Y_2^N | W_2 ) -  H( V_2^N | Y_2^N) \right \} \\
& \overset{(b)}{=}    H(Y_1^{N} | V_1^N )  + H(Y_2^{N} | V_2^N)  \\
&  + \left\{ I (V_1^N; W_2) + H(V_2^N | W_1) - H( Y_1^N | W_1 ) - H( V_1^N | Y_1^N)  \right \} \\
& + \left\{ I  (V_2^N; W_1) + H(V_1^N |W_2) - H( Y_2^N | W_2 ) - H( V_2^N | Y_2^N)  \right \} \\
&\overset{(c)}{\leq} H(Y_1^{N} | V_1^N )  + H(Y_2^{N} | V_2^N)  + H(\tilde{Y}_1^N|W_1) + H(\tilde{Y}_2^N|W_2) \\
&\overset{(d)}{\leq} \sum \left \{ H(Y_{1i} | V_{1i})  + H(Y_{2i} | V_{1i} )  + H(\tilde{Y}_{1i}) + H(\tilde{Y}_{2i}) \right \} \\
&\overset{(e)} \leq 2N \max (n-m,m) + 2N \lambda \max (\tilde{n}, \tilde{m})
\end{split}
\end{align*}
where $(a)$ and $(b)$ follow from a chain rule; $(c)$ follows from Claim~\ref{claim:blancebound} (see below); $(d)$ follows from the fact that conditioning reduces entropy; $(e)$ follows from
$\sum H(\tilde{Y}_{ki}) \leq N \lambda \max(\tilde{n},\tilde{m})$.
Therefore, we get the desired bound.

\begin{claim}
\label{claim:blancebound}
For $(k,l) = (1,2)$ or $(k,l) = (2,1)$,
\begin{align*}
\begin{split}
& I (V_k^N; W_l)  +  H(V_l^N | W_k)  -  H( Y_k^N | W_k )  -   H( V_k^N | Y_k^N)  \\
&\leq H(\tilde{Y}_k^N |W_k).
\end{split}
\end{align*}
\end{claim}
\begin{proof}
By symmetry, it is enough to prove only one case.
\begin{align*}
\begin{split}
 & I (V_1^N; W_2) + H(V_2^N | W_1) - H( Y_1^N | W_1 ) - H( V_1^N | Y_1^N)  \\
 & \overset{(a)}{=}  I (V_1^N; W_2) + \left \{ H(V_2^N | W_1, \tilde{Y}_1^N ) - H( Y_1^N | W_1, \tilde{Y}_1^N ) \right \} \\
&\quad +  I ( V_2^N; \tilde{Y}_1^N | W_1) - I ( Y_1^N; \tilde{Y}_1^N | W_1)    - H( V_1^N | Y_1^N) \\
 & \overset{(b)}{=}   I (V_1^N; W_2) +   H( \tilde{Y}_1^N | W_1, Y_1^N) \\
 & \quad - H( \tilde{Y}_1^N | W_1, V_2^N)  - H( V_1^N | Y_1^N) \\
 & \overset{(c)}{\leq} I (V_1^N; W_2) +   H( \tilde{Y}_1^N | W_1, Y_1^N, V_1^N)  \\
 &\quad + H( V_1^N | W_1, Y_1^N ) - H( V_1^N | Y_1^N)    \\
 & \overset{(d)}{\leq} I (V_1^N; W_2) +   H( \tilde{Y}_1^N | W_1, V_1^N  )  \\
 & \leq   I (W_1, V_1^N; W_2) +   H( \tilde{Y}_1^N | W_1, V_1^N)  \\
 & =   H (V_1^N |W_1) +   H( \tilde{Y}_1^N | W_1, V_1^N)  =  H (\tilde{Y}_1^N, V_1^N | W_1)     \\
& \overset{(e)}{=}    H (\tilde{Y}_1^N | W_1)
\end{split}
\end{align*}
where $(a)$ follows from a chain rule; $(b)$ follows from $H(Y_{1}^N | W_1, \tilde{Y}_1^{N} ) = H(V_{2}^N | W_1, \tilde{Y}_1^{N})$ due to the fact that $X_{1}^N$ is a function of $(W_1,\tilde{Y}_1^{N-1})$; $(c)$ follows from a chain rule and the fact that entropy is non-negative; 
$(d)$ follows from the fact that conditioning reduces entropy; $(e)$ follows from the fact that $V_1^N$ is a function of $(W_1, \tilde{Y}_1^{N-1})$.
\end{proof}

\section{Two-way IC with Four Messages}
\label{sec:4messageTwo-wayIC}
Motivated by the fact that feedback can provide a net capacity gain, we now explore the role of interaction in a more general scenario where there are four messages in total: two forward messages; and two additional independent backward messages from user $\tilde{k}$ to user $k$, $k=1,2$.
In this scenario, the encoded signal $\tilde{X}_{ki}$ of user $\tilde{k}$ is now a function of $(\tilde{W}_k, Y_k^{i-1})$ instead of $Y_k^{i-1}$ only.





We will demonstrate from an example in Fig.~\ref{fig:4message_example} that interaction can improve the non-interactive rate significantly. For simplicity, we focus on a sum-rate pair of the forward and backward messages, denoted by $(R_{\sf sum}, \tilde{R}_{\sf sum}):= (R_1 +R_2, \tilde{R}_1 + \tilde{R}_2)$. Note that the non-interactive  capacity region is $\{ (R_{\sf sum}, \tilde{R}_{\sf sum}): R_{\sf sum} \leq 2, \tilde{R}_{\sf sum}=0 \}$~\cite{ElGamal:it82,bresler:europe}. On the other hand, we will show that interaction gives:
\begin{align}
\{ (R_{\sf sum}, \tilde{R}_{\sf sum}): \tilde{R}_{
\sf sum} \leq 1, R_{\sf sum} +  \tilde{R}_{\sf sum} \leq 3 \}.
\end{align}
Note that interaction provides an arbitrarily large gain in capacity. $\tilde{R}_{\sf sum}$ can be increased up to 1 from 0, which implies an $\infty \%$ improvement. With Type-I and Type-II schemes in Section~\ref{sec:achievability}, we can achieve the $(0,1)$ and $(3,0)$ points respectively. On the other hand, a new idea emerges to achieve a corner point of $(2,1)$.

\begin{figure}[t]
\begin{center}
{\epsfig{figure=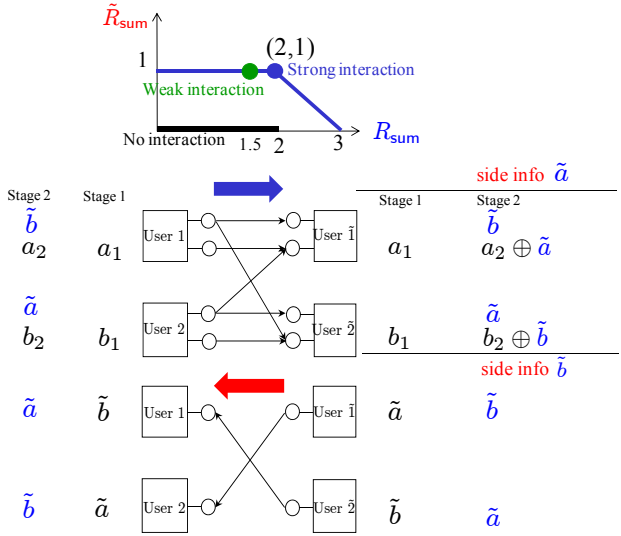, angle=0, width=0.45\textwidth}}
\end{center}
\caption{Two-way interference channel with four messages. Illustration of an achievable scheme for $(R_{\sf sum},\tilde{R}_{\sf sum})=(2,1)$.} \label{fig:4message_example}
\end{figure}

The example in Fig.~\ref{fig:4message_example} shows an achievable scheme for the $(2,1)$ point. Here we will demonstrate that during two stages, user 1 and 2 can send $(a_1,a_2)$ and $(b_1, b_2)$ respectively, while user $\tilde{1}$ and $\tilde{2}$ can transmit $\tilde{a}$ and $\tilde{b}$ respectively. In the first stage, user 1 and 2 send $a_1$ and $b_1$ using its own private level respectively. Meanwhile user $\tilde{1}$ and $\tilde{2}$ send $\tilde{a}$ and $\tilde{b}$ respectively through the backward IC. User 1 then gets the unwanted information $\tilde{b}$ and similarly user 2 receives $\tilde{a}$. In the second stage, through the forward IC, user 1 feeds $\tilde{b}$ back to user $\tilde{1}$ using the top level, and similarly user 2 feeds $\tilde{a}$ back to user $\tilde{2}$. Here the key observation is that this feedback transmission comes \emph{for free}, i.e., it does not hurt forward-message transmission of $(a_2,b_2)$. Notice that $\tilde{a}$ and $\tilde{b}$ are user $\tilde{1}$'s and $\tilde{2}$'s own information respectively. This allows user 1 and 2 to send their own forward information $a_2$ and $b_2$ without being interfered. In other words, exploiting $\tilde{a}$ as side information, user $\tilde{1}$ can decode $a_2$, and similarly user $\tilde{2}$ can decode $b_2$. Upon receiving $(\tilde{b},\tilde{a})$, user $\tilde{1}$ and $\tilde{2}$ transmit the other user's information respectively, thus enabling user 1 and 2 to decode their desired signals. Therefore, we can achieve $(2,1)$.

{\bf Mixing Forward-and-Backward Messages:}
Interestingly, this interactive scheme includes the mixture of forward-and-backward messages. Note that in the second stage, user 1 sends  $(\tilde{b},a_2)$ at the same time. We say that interaction is \emph{strong} if the mixture is allowed. On the other hand, if the mixture is not allowed (that we call \emph{weak} interaction), the performance is degraded. For example, we can show that given the constraint of $\tilde{R}_{\sf sum}=1$, the weak interaction provides at most 1.5 bits for $R_{\sf sum}$. See Appendix~\ref{appen:proofCweak} for the proof.
 This shows that interaction can provide larger gain when allowing for the mixture of different messages.

\section{Discussion}
\label{section:Comparison}

\subsection{Net Feedback Gain}
\begin{figure}[t]
\begin{center}
{\epsfig{figure=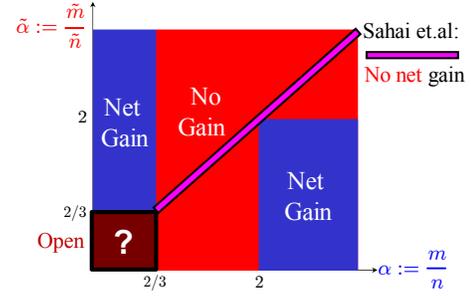, angle=0, width=0.33\textwidth}}
\end{center}
\caption{Two-way IC with two forward messages: Net gain.} \label{fig:netgainregime}
\end{figure}

Using Theorems~\ref{thm:achievability} and~\ref{thm:outerbound}, we identify channel regimes where feedback can provide a net gain in capacity. See Fig.~\ref{fig:netgainregime}. In fact, Sahai~\emph{et.al.}~\cite{Sahai:ITW09} also considered the two-way ADT deterministic IC with two forward messages, and showed that for $\alpha=\tilde{\alpha} \geq \frac{2}{3}$, there is no net feedback gain. In this work, we consider arbitrary forward-and-backward ICs. As a result, for the regimes of $\frac{2}{3} \leq \alpha \leq 2$ and $(\alpha \geq 2, \tilde{\alpha} \geq 2 )$, we obtain the similar result: there is no net feedback gain. For $(\alpha < \frac{2}{3}, \tilde{\alpha} > \frac{2}{3} )$ or $(\alpha > 2, \tilde{\alpha} < 2 )$, however, we show that feedback can provide a net gain in capacity. Here we say that net feedback gain occurs if there exists $\lambda$ such that the capacity gain due to feedback is strictly larger than the capacity gain due to independent-message transmission. For the remaining regime $(\alpha <\frac{2}{3}, \tilde{\alpha} < \frac{2}{3})$, whether or not feedback provides net gain remains open. The optimality proof of our achievable scheme with a new converse proof will show no net gain in this regime. Or enhancing Type-II or Type-III schemes (if possible) will show net feedback gain in this regime.

\subsection{Potential to Broadband Systems}

\begin{figure}[t]
\begin{center}
{\epsfig{figure=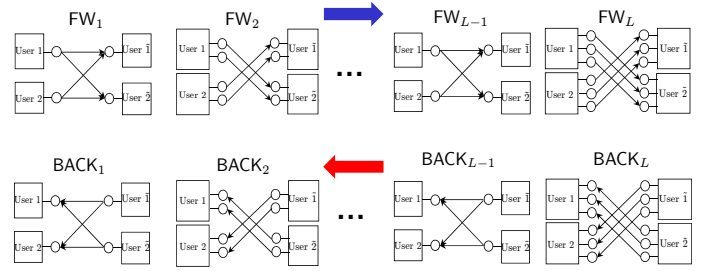, angle=0, width=0.5\textwidth}}
\end{center}
\caption{Two-way parallel IC with four messages. The rich diversity on channel gains across many parallel subchannels can often occur in broadband systems.} \label{fig:broadband}
\end{figure}

In Fig.~\ref{fig:netgainregime}, we have seen that there is net feedback gain when the forward IC well matches with the backward IC. We can also check that the gain is significant when there is strong asymmetry between $\alpha$ and $\tilde{\alpha}$. While this strong asymmetry is not likely to occur in narrowband systems, it can often occur in broadband systems where there are a multitude of subchannels with a wide dynamic range of channel gains. For example, in 3GPP-LTE and WiMAX systems, we can easily expect rich diversity on channel gains, since an operating bandwidth of the systems (around 20 MHz) is much larger than coherence bandwidth of typical wireless channels (around the order of 0.1 MHz).

Fig.~\ref{fig:broadband} illustrates an example which can represent this scenario where there are a variety of parallel subchannels. Using our results, we can see that pairs of $({\sf FW}_1, {\sf BACK}_2)$ and $({\sf FW}_2, {\sf BACK}_1)$, for instance, can provide significant gain with interaction. One more interesting observation is that even though forward-and-backward parallel subchannels are identical, there exist many pairs of forward-backward subchannels that can provide a net capacity gain. In tomorrow's communication systems, we expect a broader system bandwidth to support a variety of multimedia services. Therefore, we believe that our feedback idea will provide more significant insights into the design of future communication systems.


\section{Conclusion}
\label{sec:conclusion}

For the two-way ADT deterministic IC, we developed three types of achievable schemes and derived outer bounds, thereby establishing the sum capacity except for the regime of $(\alpha < \frac{2}{3}, \tilde{\alpha} <1)$. As a consequence, we developed a new viewpoint on the use of the backward IC: the backward IC can be more efficiently used for feedback rather than if it were used for sending its own backward messages. The gain comes from the fact that the channel use for feedback enables the exploitation of side information at forward-message-senders to make the backward IC effectively more capable. Our future work is along several new directions: (1) Extending to general channel settings; (2) Exploring the four-message two-way IC further; (2) Translating to the Gaussian channel.

\appendices

\section{Proof of $C_{\sf sum} = 1.5$ when $\tilde{R}_{\sf sum} =1$}
\label{appen:proofCweak}
For achievability, we split the forward channel into two parts: assigning the $\frac{1}{4}$ fraction of time for sending feedback signals for aiding backward-message transmission; assigning the remaining $\frac{3}{4}$ fraction of time for sending its own forward messages. Then, using Type-I scheme, we can easily achieve $(1.5,1)$. For converse, we split each channel into two parts and then use Theorem~\ref{thm:outerbound}. Similar to $\lambda$, we define $\tilde{\lambda}$ as the fraction of time that the forward channel uses for feedback. Using Theorem~\ref{thm:outerbound}, we then get:
\begin{align*}
R_{\sf sum} &\leq \min \left\{ 2 (1 - \tilde{\lambda}) + 2 \lambda, 3 (1 - \tilde{\lambda}) \right \}, \\
\tilde{R}_{\sf sum} &\leq \min \left\{ 4 \tilde{\lambda}, 1 - \lambda \right \}.
\end{align*}
Since $\tilde{R}_{\sf sum}=1$, $\lambda=0$ and $\tilde{\lambda}=\frac{1}{4}$. Therefore, $R_{\sf sum} \leq 1.5$.
%

\bibliographystyle{ieeetr}
\bibliography{bib_feedback}

\end{document}